\documentclass[10pt]{amsart}
\usepackage{inputenc}
\usepackage[T2A]{fontenc}
\inputencoding{cp1251}

\usepackage{enumerate}
\usepackage{amsmath,amsthm,amssymb}
\usepackage{amsfonts}
\usepackage{amsthm}
\usepackage{epsfig}

\newtheorem{thm}{Theorem}[section]
\newtheorem{lem}{Lemma}[section]
\newtheorem{cor}{Corollary}[section]

\newtheorem{rem}{Remark}[section]

\theoremstyle{definition}
\newtheorem*{Def*}{Definition}
\newtheorem*{rems*}{Remarks}
\numberwithin{equation}{section}

\newcommand{\D}{\displaystyle}
\newcommand{\be}{\begin{equation}}
\newcommand{\ee}{\end{equation}}

\newcommand{\ii}{\textup{i}}
\newcommand{\la}{\lambda}

\DeclareMathOperator{\res}{res}
\begin{document}

\title[Finite-gap solutions of the Sine-Gordon equation]{Finite-gap solutions of the Sine-Gordon equation
}

\author{ V.P. Kotlyarov}
\address{Institute for Low Temperature Physics\\ 47,Lenin ave\\ 61103 Kharkiv\\ Ukraine}
\email{kotlyarov@ilt.kharkov.ua}

\begin{abstract}
This paper contains first results on the finite-gap integration of  the Sine-Gordon equation. They were published on Russian  in 1976.
The papers \cite{Koz}, \cite{KK}, \cite{KK02} have been rewritten in the English language with small modifications for a convenience. Such a translation was made due to requests of some interested readers.

In those papers, the method of constructing of the finite-gap solutions of the equation $u_{tt}-u_{xx}+\sin u=0$ was proposed. The explicit formulae were obtained for these solutions. The formulae are constructed in terms of $\theta$-functions and they are analogous to the formulae obtained by A.R.Its and V.B.Matveev \cite{IM}, B.A.Dubrovin and S.P.Novikov \cite{DN} for periodic and almost periodic solutions to the Korteweg de Vries equation.
\end{abstract}


\maketitle


{\it \hfill This reprint on English language is dedicated  to the memory of V.A. Kozel}
\section{Introduction}

V.A.Marchenko \cite{M} and S.P.Novikov \cite{N} have developed
two different approaches  to solving of the periodic problem for the Korteweg de Vries equation. Their results have generated a series of papers devoted to the periodic problem (see review \cite{DMN}). In particular, papers written by A.R.Its and V.B.Matveev \cite{IM}, B.A.Dubrovin and S.P.Novikov \cite{DN}, A.R.Its \cite{I} have presented explicit formulae for the so called finite-gap (periodic and almost periodic) solutions of KdV equation and nonlinear Schr\"odinger equation. These formulae were obtained  using  results and ideas of N.I. Akhiezer \cite{A} relating to inverse problems for the Hill operator. In all these papers an essential role have played the self-adjointness of corresponding inverse problems.

A lot of interesting problems for nonlinear equations are connected with not selfadjoint Lax operators. Unfortunately, the solvability of  a periodic inverse problem  is unknown for not selfadjoint operators. Therefore in the not selfadjoint case, we use another approach where not using spectral theory of the corresponding operators. This approach uses the method from the paper  \cite{M}.  This method allows to describe finite-gap solutions with the help of autonomous nonlinear ODEs (sf.\cite{Koz}, \cite{K}). On another hand, our approach uses ideas of   papers \cite{IM}, \cite{D} for an integration of such type autonomous nonlinear ODEs as well as deducing the explicit formulae.

In this paper we give complete results obtained by this way for the Sine-Gordon equation:
\be\label{SG}
u_{tt}-u_{xx}+\sin u=0.
\ee
The main results are published in \cite{Koz}, \cite{KK}. The same approach have been   used in \cite{K}, \cite{IK} for nonlinear Schr\"odinger equation. It is necessary to emphasize that  I.M. Krichever \cite{Kri1}, \cite{Kri2} has proposed an algebraic geometry method for a construction finite-gap solutions of those nonlinear equations which appeared in the scheme of V.E. Zakharov and A.B. Shabat  \cite{ZSh}.  For equation (\ref{SG}) a modification of the Krichever method was done by A.R. Its. Derived by A.R.Its  formulae are coincided with our formulae deduced earlier in \cite{KK}. The authors bring  their appreciation to A.R.Its and I.M.Krichever which have provided us the opportunity to know the results before their publication.

\section{Finite-gap potentials}

As it is  known \cite{ZTF}, the inverse scattering method is associated with a pair of the Lax linear operators $\hat L$ and $\hat M$ such that the formal equality $\hat L_t=[\hat L,\hat M]=\hat L\hat M-\hat M\hat L$ generates  a nonlinear equation.  In the case of the Sine-Gordone equation  (\ref{SG}) operator $\hat L$ has the form:
\be\label{L}
\hat L=I\frac{d}{dx}+V(x,t), \qquad I=\begin{pmatrix}J&0\\0&0\end{pmatrix},\qquad
V+\begin{pmatrix}A&C\\C&0\end{pmatrix}
\ee
$$
J=\begin{pmatrix}
    0 & -1 \\
    1 & 0 \\
  \end{pmatrix}, \qquad A=\frac{\ii}{4}\begin{pmatrix}
                                         0 & w \\
                                         w & 0 \\
                                       \end{pmatrix}, \qquad C=\frac{1}{4}\begin{pmatrix}
                                         e^{\ii v/2} & 0 \\
                                         0 & e^{-\ii v/2} \\
                                       \end{pmatrix},
$$
where $v=v(x,t)$ and $w=w(x,t)$ are arbitrary complex-valued functions.

Let four dimensional vector function $\varphi(x)$ is a solution of equation $L\varphi=\la\varphi$ with an arbitrary complex number $\la\neq0$. Here and below in this section potential $V$ depends on $x$ only. In view of the degeneration of operator $L$ this equation is equivalent to the following one:
\be\label{L1}
L\psi:=J\psi^\prime+A\psi+\frac{1}{\la}H\psi=\la\psi, \qquad \la\neq0,
\ee
where $H=C^2$, and $\psi=\psi(x)$ is two dimensional vector function. If we put
$B:=\la J-JA-\frac{1}{\la} JH$, then equation (\ref{L1}) takes the form:
\be\label{B}
\psi^\prime=-B\psi
\ee

\begin{Def*}\emph{ Matrix valued function $V(x)$ is called a finite-gap potential of the operator $L$, if the matrix equation
\be\label{def}
\Psi^\prime_x=[\Psi, B], \qquad [\Psi, B]=\Psi B-B\Psi
\ee
has a nontrivial solution such that $\Psi$ is a matrix polynomial with the zero trace ($\rm{Sp}\Psi=0$)}.
\end{Def*}

Such a  definition   does not use  the periodicity of matrix $V(x)$.
In the periodic case it is equivalent to the definition from the paper \cite{K} for the Dirac operator. It is also equivalent  to the necessary and sufficient conditions for the Sturm-Liouville operator to be finite-gap \cite{IM}, i.e. when the absolutely continuous spectrum of the operator consists from finite number of closed  intervals on the real axis.

To understand this definition let us use results from \cite{Koz}. Let $v(x,y,t)$ and $w(x,y,t)$ be smooth complex-valued functions for all $x,y,t\in\mathbb{R}$. Let vector $\psi=\psi(x,\la,y,t)$ be a solution of  equation (\ref{L1}) with the  matrices $A$ and $H$  given by $w:=w(x,y,t)$ and $v:=v(x,y,t)$, and $\la$ is an arbitrary nonzero complex number. In connection with  equation (\ref{L1}) we introduce a set of operators:
\be\label{M}
M_k=D_k+B_k(x,\la,y,t), \quad k=1,2\qquad B_k=\begin{pmatrix}
                                                \ii\frac{w}{4} & -\la+\frac{(-1)^k}{16\la}e^{-\ii v} \\
                                                \la-\frac{(-1)^k}{16\la}e^{\ii v} & -\ii\frac{w}{4} \\
                                              \end{pmatrix},
\ee
where $D_1$ and $D_2$ are differentiations  with respect to  $t$ and $y$. Then it is easy to verify that the following equality
\be\label{LMpsi}
LM_k\psi-\la M_k\psi=(S_k+\frac{1}{\la} V_k)\psi
\ee
is fulfilled for any solution $\psi$ of the equation (\ref{L1}).
The matrices $S_k$ and $V_k$ have the form:
\begin{align*}
S_1=&-D_1 A+A^\prime_x+2(HJ-JH),\\
V_1=&-D_1 H-H^\prime_x-2(HJA-AJH),\\
S_2=&-D_2 A+A^\prime_x,\\
V_2=&-D_2 H+H^\prime_x.
\end{align*}
Equality (\ref{LMpsi}) together with (\ref{L}) show that the operator $M_1$ (and $M_2$) transforms the solution $\psi$ of equation (\ref{L1}) into a solution of the same equation (\ref{L}) iff the functions $v(x,y,t)$ and $w(x,y,t)$ satisfy the system of equations:
\be\label{uwSG}
w_t^\prime-w_x^\prime+\sin v=0, \qquad v_t^\prime+v_x^\prime-w=0
\ee
(and
\be\label{uwshift}
w_x^\prime-w_y^\prime=0,\qquad v_x^\prime-v_y^\prime=0.)
\ee
The system (\ref{uwSG}) is equivalent to the sine-Gordon equation (\ref{SG}) for the function $v(x,y,t)$, where $y$ is a parameter while the equations (\ref{uwshift}) give that
$$
v(x,y,t)=v_0(x+y,t),\qquad w(x,y,t)=w_0(x+y,t),
$$
where $v_0$ and $w_0$ are some smooth functions.
On the other hand, if $u(x,t)$ ($x,t\in\mathbb{R}$) satisfies (\ref{SG}) then the function $u(x+y,t)$ solves equation (\ref{SG}) for all $y\in\mathbb{R}$. Let us put
\be\label{v=u}
v(x,y,t)=u(x+y,t),\qquad w(x,y,t)=u_t^\prime(x+y,t)+u_x^\prime(x+y,t)
\ee
This pair of functions satisfy  (\ref{uwSG}), (\ref{uwshift}) simultaneously. Therefore, the both of operators $M_1$, $M_2$ transform the solution of  (\ref{L1}) into a solution of the same equation (\ref{L1}).

Introduce now the fundamental matrix of equation (\ref{L1}) in  such
a way that
$$
\Psi(x,y,t, \la)=\begin{pmatrix}
                   \psi_{11}(x,y,t, \la) & \psi_{12}(x,y,t, \la) \\
                   \psi_{21}(x,y,t, \la)& \psi_{22}(x,y,t, \la) \\
                 \end{pmatrix}, \qquad \Psi (x,y,t, \la)\vert_{x=0}\equiv I
$$
\begin{thm}
Let $u(x,t)$ ($x,t\in\mathbb{R}$) be a solution of (\ref{SG}).Then, for all $x,y,t\in\mathbb{R}$, the fundamental matrix of equation (\ref{L1}) satisfies the following equations:
\be\label{DkPsi}
D_k\Psi=\Psi A_k(y,t,\la)- A_k(x+y,t,\la)\Psi, \qquad k=1,2,
\ee
where $ A_k(x+y,t,\la)= B_k(x,y,t,\la)=B_k(0,x+y,t,\la)$ are defined by (\ref{M}) with $v$ and $w$ from (\ref{v=u}).
\end{thm}
\begin{proof}
Any solution of equation (\ref{L1}) is represented as $\Psi(x,y,t,\la)\psi_0$, where $\psi_0$ is an  initial condition. As $u$ satisfies equation (\ref{SG}), then  $M_k$ transforms the solution of  equation(\ref{L1}) into another one's. Hence vectors $M_k\Psi(x,y,t,\la)\psi_0$ ($k=1,2$) is also a solution of the same equation and
$$
M_k\Psi(x,y,t,\la)\psi_0=\Psi(x,y,t,\la)\{M_k\Psi(x,y,t,\la)\psi_0\}\vert_{x=0}.
$$
Since $\{M_k\Psi(x,y,t,\la)\psi_0\}\vert_{x=0}=\{M_k\}\vert_{x=0}\Psi(0,y,t,\la)\psi_0=B_k(0,y,t,\la)\psi_0$
then $\{M_k\Psi-\Psi B_k(0,y,t,\la) \}\psi_0=0$ and, due to the arbitrariness of $\psi_0$, one find that
$$
M_k\Psi(x,y,t,\la)-\Psi(x,y,t,\la) B_k(0,y,t,\la)=0.
$$
The last equality  is the matrix equation (\ref{DkPsi}).
\end{proof}

\begin{cor}
Let $u(x,t)$  ($u(x+l,t)=u(x,t)$)  is a periodic solution of (\ref{SG}). Then equations (\ref{B}), (\ref{v=u}), (\ref{DkPsi}) give $A_k(y+l,t,\la)=A_k(y,t,\la)$ and
\be\label{mon}
D_k\hat\Psi=\hat\Psi A_k(y,t,\la)- A_k(y,t,\la)\hat\Psi,\qquad k=1,2,
\ee
where $\hat\Psi(y,t,\la):=\Psi(l,y,t,\la)$ is the monodromy matrix of equation (\ref{L1}). Moreover, if we put $y=x$ and $k=2$ then equation (\ref{mon}) coincides with equation (\ref{def}) which defines finite-gap potentials.
\end{cor}

Now we return to the finite-gap potentials. Let $\psi_{ik}$ ($i,k=1,2$) be entries of the matrix $\Psi$. Then matrix equation (\ref{def}) is equivalent to the following system of differential equations:
\begin{align}\label{sysx}
\ii f^\prime=&\left(\la-\frac{e^{\ii u}}{16\la} \right)\psi+\left(\la-\frac{e^{-\ii u}}{16\la} \right)\varphi,\nonumber\\
\ii\psi^\prime=&\frac{w}{2}\psi+2\left(\la-\frac{e^{-\ii u}}{16\la} \right)f,\\
\ii\varphi^\prime=&-\frac{w}{2}\varphi+2\left(\la-\frac{e^{\ii u}}{16\la} \right)f,\nonumber\\
g^\prime=&0,\nonumber
\end{align}
where $2g:=\psi_{11}+\psi_{22}$ and $2\ii f:=\psi_{11}-\psi_{22}$, $\psi:=\psi_{12}$, $\varphi:=\psi_{21}$. Prime $\prime$ denote the differentiation with respect to $x$. It is obviously that any  matrix, which is polynomial in $\la$ and satisfies  equation (\ref{def}), uniquely defines a  polynomial solution of the system (\ref{sysx}) and vice versa.

Following to \cite{Koz} we look for a polynomial solution of the system (\ref{sysx})
in the form:
\be\label{pol}
f=\sum\limits_{k=1}^{N_1} f_k \la^{2k-1}, \quad
\psi=\sum\limits_{k=0}^{N_2} \psi_k \la^{2k}, \quad
\varphi=\sum\limits_{k=0}^{N_3} \varphi_k \la^{2k}.
\ee
The necessary and sufficient conditions for polynomials (\ref{pol}) to be a solution of (\ref{sysx}) are as follows:
\begin{itemize}
\item $N_1=N_2=N_3=N, \qquad N=0,1,2,3,\ldots$;\\
\item $\psi_N=-\varphi_N$ is independent on $x$;\\
\item $e^{\ii u}\psi_0 + \varphi_0e^{-\ii u}=0$;\\
\item $w=-4\D\frac{f_N}{\psi_N}$;\\
\item polynomial coefficients are satisfy the nonlinear autonomous system of differential equations:
\begin{align}
\ii f_m^\prime=&\psi_{m-1} +\varphi_{m-1}-\frac{\psi_{0}\varphi_{m}-\varphi_{0}\psi_{m}}
{16\sqrt{-\psi_{0}\varphi_{0}}}\nonumber\\
\ii\psi^\prime_{m}=&2f_{m}  -\frac{\psi_{0} f_{m+1 }}
{8\sqrt{-\psi_{0}\varphi_{0}}}-\frac{2f_{N}\psi_{m} }{\psi_{N} },\quad m=0,1,\ldots,N \label{sysmx}\\
\ii\varphi^\prime_{m}=&2f_{m} +\frac{\varphi_{0} f_{m+1} }
{8\sqrt{-\psi_{0}\varphi_{0}}}+\frac{2f_{N} \varphi_{m}}{\psi_{N}}.\nonumber
\end{align}
\end{itemize}
Any solution of the system (\ref{sysmx}) generates the matrix polynomial solution  $\Psi$ of (\ref{def}). Such solution satisfies the symmetry property:
\be\label{sym}
\sigma_3\Psi(x,-\la)\sigma_3=-\Psi(x,\la), \qquad \sigma_3=\begin{pmatrix}
                                                    1 & 0\\
                                                    0& -1 \\
                                                  \end{pmatrix}.
\ee
Thus,  necessary and sufficient conditions ( pointed above) give a complete description of all finite-gap potentials, i.e. when equation (\ref{def}) has a polynomial in $\la$ solution of the form (\ref{pol}). It is easy to prove that other finite-gap potentials do not exist. Indeed, let $V(x)$ be an arbitrary,  finite-gap potential. Then by the definition (\ref{def}) there  exists a polynomial in $\la$ solution $\Psi(x,\la)$. In view of the symmetry property
$$
\sigma_3B(x,-\la)\sigma_3=B(x,\la),
$$
the polynomial $\sigma_3\Psi(x,-\la)\sigma_3$ is also a solution of (\ref{def}). Hence the matrix $\Psi(x,\la)-\sigma_3\Psi(x,-\la)\sigma_3$ is a polynomial solution of (\ref{def}) which satisfies the symmetry property (\ref{sym}). The above consideration lead to a more precise definition.
\begin{Def*}\emph{
Matrix  $V(x)$ is called N-gap  potential if equation (\ref{def}) has a polynomial in $\la$ solution of the form (\ref{pol}), where $N_1=N_2=N_3=N$ and, at the same time, there does not exist a polynomial in $\la$ solution of the degree is less than $N$.}
\end{Def*}
The mentioned above necessary and sufficient conditions allow to describe the set of finite-gap potentials $V(x)$ of the operator $L$, and also to find all such potentials by solving autonomous system (\ref{sysmx}).

\section{Autonomous systems and sine-Gordon equation }

In this section we show that the  systems  of equations (\ref{mon}) $k=1,2$  generate some nonlinear autonomous equations, which lead to a solution of the Sine-Gordon equation (\ref{SG}). Let coefficients $u$ and $w$ of equations (\ref{mon}) $k=1,2$ are completely arbitrary, i.e. we will consider the next system of equations:
\begin{align}\label{tmon1}
\ii\dot f =&\left(\la+\frac{e^{\ii u}}{16\la} \right)\psi+\left(\la+\frac{e^{-\ii u}}{16\la} \right)\varphi,\nonumber\\
\ii\dot\psi=&\frac{w}{2}\psi+2\left(\la+\frac{e^{-\ii u}}{16\la} \right)f,\\
\ii\dot\varphi=&-\frac{w}{2}\varphi+2\left(\la+\frac{e^{\ii u}}{16\la} \right)f \nonumber
\end{align}
and
\begin{align}\label{xmon1}
\ii f^\prime=&\left(\la-\frac{e^{\ii u}}{16\la} \right)\psi+\left(\la-\frac{e^{-\ii u}}{16\la} \right)\varphi,\nonumber\\
\ii\psi^\prime=&\frac{w}{2}\psi+2\left(\la-\frac{e^{-\ii u}}{16\la} \right)f,\\
\ii\varphi^\prime=&-\frac{w}{2}\varphi+2\left(\la-\frac{e^{\ii u}}{16\la} \right)f \nonumber
\end{align}
where $u=u(x,t),\, w=w(x,t)$ are some arbitrary complex valued functions, and the dot and prime mean differentiation in $t$ and $x$ respectively (from here and below we use $x$ instead of $y$).

First of all we emphasize that  the  both of system has the conservation law:
\be\label{P}
f^2(x,t,\la)-\psi(x,t,z\la)\varphi(x,t,\la)\equiv P(\la),\qquad \la\in\mathbb{C},
\ee
where $P(\la)$ is defined by initial conditions. If $f, \psi, \varphi$ satisfy
(\ref{tmon1})  and (\ref{xmon1}) simultaneously, then $P(\la)$ is independent on $t$ and $x$. Below we consider the case when  $u,w$ are unknown functions. In this case the systems of equations are underdetermined (not closed). However, we can make them to be closed. To this end, let us seek polynomial in $\la$ solutions of these systems. It is easy to analyze that polynomial solution must be of the form:
\be\label{pol1}
f(x,t,\la)=\sum\limits_{j=1}^{N}f_j(x,t)\la^{2j-1},
\quad
\psi(x,t,\la)=\sum\limits_{j=0}^{N}\psi_j(x,t)\la^{2j}, \quad
\varphi(x,t,\la)=\sum\limits_{j=0}^{N}\varphi_j(x,t)\la^{2j}.
\ee
Let us put $f_0=0$ and $f_j, \psi_l, \varphi_m$ are equal zero when $j, l, m <0$ and $j, l, m >N$.  Then we have algebraic relations:
\begin{align*}
\psi_{N} +\varphi_{N}=&0,\\
e^{\ii u}\psi_0+e^{-\ii u}\varphi_0=&0,\\
4 f_{N}-w\varphi_{N}=&0,\\
4f_N+w\psi_{N}=&0.
\end{align*}
The last relations give unknown functions:
\be\label{uw}
e^{-\ii u}=\frac{\psi_0}{\sqrt{-\psi_{0}\varphi_{0}}},\qquad w=4\frac{f_N}{\varphi_N}=-4\frac{f_N}{\psi_N},
\ee
System (\ref{tmon1}) transforms into the system of equations on polynomial coefficients:
\begin{align}
\ii\dot f_m =&\psi_{m-1} +\varphi_{m-1}+\frac{\psi_{0}\varphi_{m}-\varphi_{0}\psi_{m}}
{16\sqrt{-\psi_{0}\varphi_{0}}}\nonumber\\
\ii\dot\psi_{m}=&2f_{m}  -\frac{\psi_{0} f_{m+1 }}
{8\sqrt{-\psi_{0}\varphi_{0}}}-\frac{2f_{N}\psi_{m} }{\psi_{N} },\quad m=0,1,\ldots,N \label{tmon2}\\
\ii\dot\varphi_{m}=&2f_{m} +\frac{\varphi_{0} f_{m+1} }
{8\sqrt{-\psi_{0}\varphi_{0}}}+\frac{2f_{N} \varphi_{m}}{\psi_{N}}.\nonumber
\end{align}
Moreover, the second and third equations of system (\ref{tmon2}) gives that $\psi_{N}$ and $\varphi_{N}$ is independent on $t$. The conservation law (\ref{P}) give that the quantity $\psi_0\varphi_0=P(0)$ is also constant value.
Substituting $e^{\ii u}, \, e^{-\ii u}, \, w$ into (\ref{tmon1}) we obtain the autonomous system of ODEs:
\be\label{F}
\dot y_j=F_j(y_0, y_1, \ldots, y_M), \qquad j=0,1,\ldots,M, \quad M=3N+2,
\ee
where $y_j=f_j\, (1\le j\le N)$, $y_{N+j}=\psi_j\, (0\le j\le N)$, $y_{2N+1+j}=\varphi_j\, (0\le j\le N)$, and the right hand sides $F_j$ are at most the second degree polynomials on $y_0, y_1, \ldots, y_M$ (since denominators of $F_j$ are constant values).

In polynomial coefficients the system (\ref{xmon1}) is written in the form:
\begin{align}
\ii f_m^\prime=&\psi_{m-1} +\varphi_{m-1}-\frac{\psi_{0}\varphi_{m}-\varphi_{0}\psi_{m}}
{16\sqrt{-\psi_{0}\varphi_{0}}}\nonumber\\
\ii\psi^\prime_{m}=&2f_{m}  -\frac{\psi_{0} f_{m+1 }}
{8\sqrt{-\psi_{0}\varphi_{0}}}-\frac{2f_{N}\psi_{m} }{\psi_{N} },\quad m=0,1,\ldots,N \label{xmon2}\\
\ii\varphi^\prime_{m}=&2f_{m} +\frac{\varphi_{0} f_{m+1} }
{8\sqrt{-\psi_{0}\varphi_{0}}}+\frac{2f_{N} \varphi_{m}}{\psi_{N}}.\nonumber
\end{align}
The function $u$ and $w$ are found by the same formulas (\ref{uw}) as before.  Again we have that $\psi_N=-\varphi_N $ and $\psi_0\varphi_0$ are independent on $x$. After the elimination of $u$ and $w$ these equations are also written in the form of autonomous system of ODEs:
\be\label{Phi}
y_j^\prime=\Phi_j(y_0,y_1, \ldots, y_M), \qquad j=0,1,\ldots,M, \quad M=3N+2,
\ee
where $\Phi_j$ are at most the second degree polynomials on $y_0, y_1, \ldots, y_M$. Let us note that without loss of a generality we can put: $\psi_N=-\varphi_N=1$.

Thus, the requirement of the existence of the polynomial in $z$ solution to systems (\ref{tmon1}) and (\ref{xmon1}) defines uniquely the functions $e^{\ii u}, \, e^{-\ii u}, \, w$ through solutions of autonomous systems (\ref{F}) and (\ref{Phi}). On the other hand, if $f_j$, $\psi_l$ and $\varphi_m$ are a solution of autonomous systems (\ref{F}), (\ref{Phi}), and $e^{\ii u}, \, e^{-\ii u}, \, w$ are the function constructed by (\ref{uw})  then polynomials (\ref{pol1}) are a  solution of systems (\ref{tmon1}), (\ref{xmon1}).
In what follows the system (\ref{tmon2}), (\ref{xmon2}) together with  (\ref{uw}), and the system (\ref{F}), (\ref{Phi}) are identified.

It is well known that necessary and sufficient conditions for the systems
(\ref{F}), (\ref{Phi}) to be compatible are the following ones:
\be\label{Frobenius}
\sum\limits_{j=0}^{N}\left(\frac{\partial F_j}{\partial y_m}\Phi_m-\frac{\partial \Phi_j}{\partial y_m}F_m  \right)=0, \qquad j=0,1,\ldots, M.
\ee
The direct calculation shows that compatibility conditions (\ref{Frobenius}) are fulfilled.
\begin{thm}
Let $f_j(x,t), \psi_l(x,t), \varphi_m(x,t)$ be a compatible local solution of autonomous system (\ref{F}), (\ref{Phi}). Then, functions $u(x,t)$ and $w(x,t)$ defined by (\ref{uw})  are a local infinitely differentiable in $x$ and $t$ solution of the following nonlinear  equations:
$$
 \dot u+u^\prime -w=0,\qquad \dot w- w^\prime+\sin u=0,
$$
\end{thm}

\begin{proof}
Since right hand sides $F_j(y_0,y_1,\ldots,y_M)$ and $\Phi_j(y_0,y_1,\ldots,y_M)$ of systems (\ref{F}), (\ref{Phi}) are infinitely differentiable in $y_0,y_1,\ldots,y_M$ then their solutions are also infinitely differentiable in $t$ and $x$.
Systems (\ref{tmon2}), (\ref{xmon2}) give:
\begin{align}
\ii(\dot f_N-f^\prime_N)=&-\frac{(\psi_0+\varphi_{0})\psi_N}
{8\sqrt{-\psi_{0}\varphi_{0}}},\label{psit}\\
\ii(\dot\psi_0+\psi^\prime_N)=&-4f_N\psi_{0}\label{psix}
\end{align}
Since
$$
f_N=-\frac{w\psi_N}{4},\qquad e^{-\ii u}=\frac{\psi_0}{\sqrt{-\psi_{0}\varphi_{0}}}
$$
then we have
$$
\dot w- w^\prime+\sin u=0, \qquad \dot u+u^\prime -w=0.
$$
As $u$ and $w$ are smooth then last equations are equivalent to the SG equation (\ref{SG}):
$$
  \ddot u-u^{\prime\prime}+\sin u=0.
$$
\end{proof}

Under some conditions systems (\ref{tmon2}), (\ref{xmon2}) have a global solution.
\begin{lem}
Let initial conditions satisfy the symmetry properties:
\be\label{sym1}
f_j=\bar f_j, \qquad \varphi_j=-\bar\psi_j.
\ee
Then system  (\ref{tmon2}), as well as  the system (\ref{xmon2}), has a solution defined for all $t\in\mathbb{R}$ or $x\in\mathbb{R}$ respectively.  The solution is uniformly bounded and possesses the same property (\ref{sym1}) for all $t$ and $x$.
\end{lem}

\begin{proof}
The proof will be done for system (\ref{tmon2}).  We first show that the symmetry property is conserved for all $t$. Indeed, let $f_j(t)$, $\psi_l(t)$, $\varphi_m(t)$ be the solution of (\ref{tmon2}). Consider a new functions:
$\hat f_j(t)=\bar f_j(t), \, \hat\psi_l(t)=-\bar\varphi_l(t), \, \hat\varphi_m(t)=-\bar\psi_m(t)$. After complex conjugation of system (\ref{tmon2}) it is easy to see that $\hat f_j(t), \, \hat\psi_l(t), \, \hat\varphi_m(t)$ satisfy (\ref{tmon2}). Under condition $t=t_0$, in view of (\ref{sym1}), these two sets of functions are coincide. Then the uniqueness theorem for ODEs gives that they are coincide for all $t$. Further, the existence theorem for ODEs gives the local solution of  (\ref{tmon2}), which inherits the symmetry properties (\ref{sym1}). In turn these properties lead to the uniform boundedness of the solution and, hence, the solution has an extension for all $t\in\mathbb{R}$.

The uniform boundedness  is a consequence   of the conservation law (\ref{P}):
\be\label{P1}
P(\la)\equiv f^2(\la,t)-\psi(\la,t)\varphi(\la,t),
\ee
where polynomials $f, \psi, \varphi$ are constructed by $f_j(t)$, $\psi_l(t)$, $\varphi_m(t)$. In view of (\ref{sym1}), $f(\la,t)=\overline f(\bar\la,t)$ and $\varphi(\la,t)=-\overline\psi(\bar\la,t)$. Hence $P(\la)=\overline P(\bar\la)$. Therefore the polynomial $P(\la)$ can be factorized:
$$
P(\la)=P^+(\la)P^-(\la),
$$
where $P^+(\la)$ is a polynomial of $2N$-th degree. Its zeroes are located in the closed upper half plane. The second polynomial is the complex conjugated to the first one: $P^-(\la)=\overline P^+(\bar\la)$. For real $\la$ conservation law  (\ref{P1}) takes the form
$$
f^2(\la,t)+|\psi(\la,t)|^2\equiv|P^+(\la)|^2.
$$
Hence
$$
|f(\la,t)|\le |P^+(\la)|, \qquad |\psi(\la,t)|=|\varphi(\la,t)|\le |P^+(\la)|,\qquad \la\in\mathbb{R}.
$$
By using  Bernstein inequalities \cite{Levin} we find
$$
\left\vert\frac{\partial^m f(\la,t)}{\partial \la^m}\right\vert\le
\left\vert\frac{d^m P^+(\la)}{d\la^m}\right\vert, \quad m=1,2,\ldots, \quad z\in\mathbb{R}.
$$
Choosing $m=0,1,2,\ldots, N$ and $z=0$ we obtain
$$
|f_m(t)|\le|P^+_m|, \qquad |\psi_m(t)|=|\varphi_m(t)|\le |P^+_m|,
$$
where $P^+_m$ $m=1,2,\ldots,N$ are coefficients of the polynomial $P^+(\la)$, which uniquely defined by initial conditions of the corresponding autonomous system.
The analogous proof occurs for the system (\ref{xmon2}).
\end{proof}
\begin{cor}
If initial conditions for the systems (\ref{tmon2}) and (\ref{xmon2}) possess properties (\ref{sym1}) then there exists  the compatible, global, infinitely differentiable and bounded solution $\hat f_j(t)=\bar f_j(t), \, \hat\psi_l(t)=-\bar\hat\varphi_l(t), \, \hat\varphi_m(t)=-\bar\psi_m(t)$.
Moreover, the function $$u(x,t)=-\arg \psi_0(x,t)$$ is a real global ($x,t\in\mathbb{R}$) and smooth solution of the Sine-Gordon nonlinear equation (\ref{SG}).
\end{cor}
\begin{proof}
Indeed, due to the symmetry properties,
$$
u(x,t)=\ii\ln\frac{\psi_0(x,t)}{\sqrt{-\psi_{0}\varphi_{0}}}=
\ii\ln\frac{\psi_0(x,t)}{\sqrt{|\psi_{0}|^2}}.
$$
The conservation law (\ref{P1}) give that $|\psi_0(x,t)|^2=-\psi_0(x,t)\varphi_0(x,t)=|\psi_{0}|^2=P(0)$. Hence $|\psi_0(x,t)|=|\psi_0|$ for all $x$ and $t$. Therefore
$$
u(x,t)= \ii\ln\left\vert\frac{\psi_0(x,t)}{|\psi_{0}|}\right\vert-\arg\psi_0(x,t)=-\arg\psi_0(x,t)
$$
is a solution of the Sine-Gordon nonlinear equation (\ref{SG}).
\end{proof}
Thus it is proved that any compatible solution of the autonomous systems (\ref{tmon2}) and (\ref{xmon2}) generates by (\ref{uw}) one-parameter (in $t$) set of the finite-gap potentials $V(x,t)$ of the Lax operator $L$. We call the corresponding function $u(x,t)$ as a finite-gap solution of the Sine-Gordon equation.

\section{Abel map and Jacobi inversion problem}

Since $\psi_N$ and $\varphi_N$ are independent on $x$ an $t$ we put  $\psi_N=-\varphi_N=1$ without loss of generality. Then polynomial $\psi(x,t,\la)$ has the form:
$$
\psi(x,t,\la)=\prod\limits_{j=1}^N(\la-\la_j)(\la+\la_j)
$$
and hence $\psi_0(x,t)=(-1)^N\prod\limits_{j=1}^N\mu_j(x,t), \, \mu_j(x,t):= \la_j^2(x,t)$.  If we put $\la^2=z$ then, leaving  the same notation for polynomial $P$,
$$
P(z)=\sum\limits_{j=0}^{2N} p_k z^{2k},
$$
we find
\be\label{uln}
u(x,t)=\ii\ln\left(\D\frac{(-1)^N\prod\limits_{j=1}^N\mu_j(x,t)}{\sqrt{p_0}}\right).
\ee
Introduce now new functions $\hat\psi$ and $\hat f$ such that
$$
\hat\psi(x,t,z)\vert_{z=\la^2}\equiv\psi(x,t,\la),\qquad
\hat f(x,t,z)\vert_{z=\la^2}\equiv \la f(x,t,\la).
$$
Then the systems (\ref{tmon1}) and (\ref{xmon1}) take the form
\be\label{hatpsi}
\ii D_k\hat\psi=2\left(1-\D\frac{(-1)^{k+N}\prod\limits_{j=1}^N\mu_j(x,t)}{16z\sqrt{p_0}}\right)\hat f+\frac{w}{2}\hat\psi, \qquad k=1,2,
\ee
where $D_1$ and $D_2$ are differentiations with respect to $t$ and $x$. These equations are equivalent to
$$
\ii D_k\ln\hat\psi=
2\left(1-\D\frac{(-1)^{k+N}\prod\limits_{j=1}^N\mu_j(x,t)}{16z\sqrt{p_0}}\right)\D\frac{\hat f}{\hat\psi}+\frac{w}{2}, \qquad k=1,2.
$$
The multiplication on $z-\mu_l$ and passage to the limit, as $z\to\mu_l$, give
$$
D_k\mu_l=2\ii \left(1-\D\frac{(-1)^{k+N}}{16\sqrt{p_0}}\prod\limits_{j\neq l}\mu_j(x,t)\right)\D\frac{\hat f(x,t,\mu_l)}{\prod\limits_{j\neq l}(\mu_l-\mu_j)}
$$
Take into account that $\hat f(x,t,\mu_l)=\sqrt{\mu_lP(\mu_l)}$ we finally obtain
\be\label{Dkmu}
D_k\mu_l=2\ii \left(1-\D\frac{(-1)^{k+N}}{16\sqrt{p_0}}\prod\limits_{j\neq l}\mu_j(x,t)\right)\D\frac{\sqrt{\mu_lP(\mu_l)}}{\prod\limits_{j\neq l}(\mu_l-\mu_j)}, \quad k=1,2;\quad l=1,2,\ldots, N
\ee
We have supposed here that polynomial $\hat\psi$ has simple zeroes. In general, they remain simple for small $x$ and $t$, if inequalities $\mu_l(0,0)\neq\mu_j(0,0)$ took place at the origin. The symmetry conditions
$$
f(x,t,\la)=\overline f(x,t,\bar\la),\qquad
\varphi(x,t,\la)=-\overline\psi(x,t,\bar\la)
$$
mean that $P(z)$ ($z=\la^2$) is positive on the real line and, hence,  zeroes of $P(z)$ form the conjugated pairs of points on the complex plane  and/or the even number of them  lies on the negative half line.

Systems (\ref{Dkmu}) are integrated by the known Abel map \cite{D}. To this end the systems must be considered on the hyper-elliptic Riemann surface $\mathcal{R}$ of the function $\sqrt{zP(z)}$. Initial data $\mu_1(0,0), \mu_2(0,0),\ldots, \mu_N(0,0)$, as points are belonging to $\mathcal{R}$, are defined in such a way that
$$
\sqrt{\mu_j(0,0)P(\mu_j(0,0))}=\sum\limits_{l=1}^N f_l(0,0)\mu^l_j(0,0).
$$

Let $E_j$ ($j=1,2,\ldots,2N$) be zeroes of polynomial $P(z)$. Consider a realization of the Riemann surface $\mathcal{R}$ as two-sheet covering of the complex plane. The sheets are merged along of disjoint intervals  $[E_{2j-1},E_{2j}]$ $j=1,2,\ldots,N$ and the positive ray $[0,\infty)$. We choose the cycle  $a_j$ in such a way that the corresponding curve  lies on the upper sheet of $\mathcal{R}$ and encloses interval $[E_{2j-1},E_{2j}]$. The cycle $b_j$ starts on the upper bank of $[E_{2j-1},E_{2j}]$ goes along the upper sheet to the ray  $[0,\infty)$ and returns along the lower sheet to the starting point. We choose also the fixed orientation of the cycles.

Introduce on $\mathcal{R}$ basis $\{\omega_j(z)\}\vert_{j=1}^N$ of Abelian integrals of the first kind:
$$
\omega_j(z)=\int\limits_\infty^z\frac{h_(s)ds}{\sqrt{sP(s)}}, \qquad
h_j(s)=\sum\limits_{l=0}^{N-1} c_{lj}s^l, \qquad  j=1,2,\ldots,N.
$$
They are normalized by conditions:
$$
\int\limits_{a_j} d\omega_l(z)=\delta_{jl}, \qquad  j,l=1,2,\ldots,N.
$$
Then the Abel map has the form:
\be\label{AM}
\nu_j(x,t)=\sum\limits_{l=1}^{N}\omega_j(\mu_l(x,t)).
\ee
Substitution (\ref{AM}) into equations (\ref{Dkmu}) gives next equations for $\nu_j (j=1,2,\ldots,N)$:
$$
D_k\nu_j(x,t)=2\ii\sum\limits_{i=0}^{N-1}c_{ij}\sum\limits_{n=1}^N
\frac{\mu_n^i}{\prod\limits_{l\neq n}(\mu_n-\mu_l)}
\left(1-\D\frac{(-1)^{k+N}}{16\sqrt{p_0}}\prod\limits_{l\neq n}\mu_j(x,t)\right)\qquad k=1,2.
$$
It is easy to see that expression $\sum\limits_{n=1}^N
\frac{\mu_n^i}{\prod\limits_{l\neq n}(\mu_n-\mu_l)}$ can be written as
$$
\sum\limits_{n=1}^N
\frac{\mu_n^i}{\prod\limits_{l\neq n}(\mu_n-\mu_l)}=
\sum\limits_{n=1}^N \res\limits_{z=\mu_n}
\frac{z^i}{\prod\limits_{l=1}^N(z-\mu_l)}.
$$
Then the theory of residues gives
$$
\sum\limits_{n=1}^N
\frac{\mu_n^i}{\prod\limits_{l\neq n}(\mu_n-\mu_l)}=\begin{cases}
\frac{(-1)^{N+1}}{\prod\limits_{l\neq n}\mu_l}, & i=-1\\
0, & 0\le i<N-1\\ 1,& i=N-1.\end{cases}
$$
and we finally have
$$
D_k\nu_j(x,t)=2\ii\left(c_{N-1j}+\D\frac{(-1)^{k}}{16\sqrt{p_0}}c_{0j} \right).
$$
Hence
$$
\nu_j(x,t)=\alpha_j x+\beta_j t+\nu_j(0,0),
$$
where
$$
\alpha_j=2\ii\left(c_{N-1j}+\D\frac{c_{0j}}{16\sqrt{p_0}} \right),\qquad
\beta_j=2\ii\left(c_{N-1j}-\D\frac{c_{0j}}{16\sqrt{p_0}} \right).
$$

Thus we came to the Jacobi inversion problem: find the functions $\mu_j$ in such a way that they satisfy equation
\be\label{JIP}
\sum\limits_{n=1}^N\omega_j(\mu_n(x,t))=\alpha_j x+\beta_j t +\sum\limits_{n=1}^N\omega_j(\mu_n(0,0))
\ee
as the points of the Riemann surface $\mathcal{R}$, i.e. the equality have to fulfill by modulus of periods of the Abelian integrals.

The application of methods from \cite{Zver} to the inversion problem (\ref{JIP})
gives an explicit formula for the function $\sum\limits_{n=1}^N\ln\mu_n(x,t)$ in terms of the theta-functions. Then, in the agreement with (\ref{uln}), we obtain the representation for the finite-gap solution of the Sine-Gordon equation (\ref{SG}).

Let $B_{ij}$ be the entries of the matrix of the $B$-periods of the basis $\{\omega_j(z)\}$. Let us consider theta function
$$
\vartheta(\vec u)=\sum\limits_{\vec m\in\mathbb{Z}^N}
\exp\left[\pi\ii\sum\limits_{ij=1}^N B_{ij} m_i m_j+2\pi\ii
\sum\limits_{j=1}^N m_ju_j\right],
$$
where the summation is made over all integers $m_1, m_2,\ldots,m_N$.
Denote by $\hat{\mathbb{R}}$ the cut Riemann surface $\mathbb{R}$ over $a$-cycles and introduce the Riemann theta function:
\be\label{RTF}
\Theta (z)=\vartheta(e_1-\omega_1(z), e_2-\omega_2(z), \ldots, e_N-\omega_N(z)),
\ee
where
\be\label{ag}
e_j=\alpha_j x+\beta_j t+\gamma_j,\qquad
\gamma_j=\sum\limits_{n=1}^N\omega_j(\mu_n(0,0))+\frac{1}{2}\sum\limits_{i=1}^N B_{ij}-\frac{j}{2}.
\ee
It is known that function $\Theta(z)$ is regular (analytic) on $\hat{\mathbb{R}}$, and its zeroes solves the Jacobi inversion problem (\ref{JIP}) or, by other words, the zeroes of $\Theta(z)$ coincide with points $\mu_j(x,t)$.

We have supposed that polynomial $\hat\psi(0,0,z)$ has simple zeroes. It means that  all $\mu_j(0,0)$ are different and hence  $\mu_j(x,t)$ are also different for sufficiently small $x$ and $t$. Therefore, among  the   points  $\mu_j(x,t)$,  there are no conjugated pairs. In turn it means that the Riemann theta function  is not trivial, i.e. $\Theta(z)\neq 0$. (We remind that the conjugated pair of points on the Riemann surface $\mathbb{R}$ is such that projections of the points on the complex plane are the same, but the points lie on the different sheets of $\mathbb{R}$).

Now we consider the function $ \ln z$, which is the normalized Abelian integral of the third kind on the Riemann surface $\mathbb{R}$ with poles at the points $0$ and $\infty$ and residues $\pm 2$, since local variables are $\tau_0=\sqrt{z}$ and   $\tau_\infty=1/\sqrt{z}$. Its $b$-periods are equal to $2\pi\ii$ or $-2\pi\ii$ in a dependence of the orientations of  cycles. We cut $\hat{\mathcal{R}}$ along a curve $S$ which connects the points $0$ and $\infty$ in such a way that this curve does not intersect $a$ and $b$-cycles and passes away from the set of points $\{\mu_1(0,0), \mu_2(0,0),\ldots,\mu_N(0,0)\}$. Then in the domain $\hat{\mathcal{R}}\setminus S$  we chose the single valued branch of $\ln z$ by fixing its argument. Let $S^+$ ($S^-$) be the left (right) bank of the cut $S$ when one passes from $0$ to $\infty$. For sufficiently small $\varepsilon>0$ and $1/r>0$
consider the closed contour
$$
\Gamma_{\varepsilon r}=S^+_{\varepsilon r}\cup O_r\cup S^-_{\varepsilon r}\cup O_{\varepsilon},
$$
where $S^\pm_{\varepsilon r}=\{z\in S^\pm: \varepsilon<|z|<r\}$ and
$O_{\varepsilon}=\{z: |z|=\varepsilon\}$, $O_r=\{z:|z|=r\}$.
Then residues theorem gives
$$
\frac{1}{2\pi\ii}\left(\int\limits_{\,\, \partial\hat{\mathcal R}}
\,\ln z\, d\ln \Theta(z)+\int\limits_{\Gamma_{\varepsilon r}} \,\ln z\, d\ln \Theta(z)\right)=\sum\limits_{j=1}^N\ln\mu_j(x,t),
$$
where $\partial\hat{\mathcal R}$ is the boundary of the domain $\hat{\mathcal R}$.
We can evaluate the left hand side using known formulas for jumps of the Abelian integrals of the third kind and boundary relations or the Riemann theta function.

To do so let us put $\sigma_j^2=1$ $j=1,2,\ldots,N$ and $2\pi\ii\sigma_j$ be  $b$-periods of $\ln z$. Then, using well-known formulas for jumps of $\omega_j(z)$ and properties of the theta function, we find
$$
 \int\limits_{\,\, \partial\hat{\mathcal R}}
\,\ln z\, d\ln \Theta(z)= $$$$
=\sum\limits_{j=1}^N \left(\int\limits_{a_j}\,\ln z\, d\ln \Theta(z)-
\int\limits_{a_j}\,(\ln z+2\pi\ii\sigma_j)\, d\ln \Theta(z)+2\pi\ii\int\limits_{a_j}\,(\ln z+2\pi\ii\sigma_j)\, d\omega_j(z)\right),
$$
i.e.
$$
\frac{1}{2\pi\ii}\int\limits_{\,\, \partial\hat{\mathcal R}}
\,\ln z\, d\ln \Theta(z)=\sum\limits_{j=1}^N \int\limits_{a_j}\,\ln z\, d\omega_j(z)+2\pi\ii\sum\limits_{j=1}^N\sigma_jm_j=\sum\limits_{j=1}^N \int\limits_{a_j}\,\ln z\, d\omega_j(z)+ 2\pi\ii k,
$$
where $m_j$ and $k$ are some integer. Taking into account  that the function $\Theta(z)$ is continuous for $z\in S$ and  relation
$$
\ln z\vert_{z\in S^-}-\ln z\vert_{z\in S^+}=4\pi\ii,
$$
we can write
$$
\frac{1}{2\pi\ii}\int\limits_{\,\, \Gamma_{\varepsilon r}}
\,\ln z\, d\ln \Theta(z)=-2 \int\limits_{\,\, S^+_{\varepsilon r}}
 d\ln \Theta(z)+$$$$+\frac{1}{\pi\ii}\int\limits_{|\tau|=\sqrt{\varepsilon}}
\,\ln \tau\, d\ln \Theta_0(\tau)-\frac{1}{\pi\ii}\int\limits_{|\tau|=1/\sqrt{r}}
\,\ln \tau\, d\ln \Theta_\infty(\tau),
$$
where $\Theta_0(\tau)$ and $\Theta_\infty(\tau)$ are elements of the function $\Theta(z)$ in neighborhoods of points $0$ and $\infty$ on $\mathcal{R}$. Since
$\Theta(z)$ is regular everywhere on $\mathcal{R}$ and does not equal to zero at the points $0$ and $\infty$, then the integrals over $|\tau|=\sqrt{\varepsilon}$ and
$|\tau|=1/\sqrt{r}$ tend to zero  as $\varepsilon\to0$ and $r\to0$. Hence
\be\label{310}
\sum\limits_{j=1}^N\ln\mu_j(x,t)=2\ln\frac{\Theta(0)}{\Theta(\infty)}+\sum\limits_{j=1}^N \int\limits_{a_j}\,\ln z\, d\omega_j(z)+ 2k\pi\ii.
\ee
At the infinity we have that all  $\omega_j(z)$ vanish while at the origin $\omega_j(0)= \sigma_j/2$. The last follows from the Riemann relations \cite{Zver}:
$$
2\pi\ii\omega_j(0)=2\pi\ii\int\limits_\infty^0d\omega_(z)=\frac{1}{2}\int\limits_{a_j} d\ln z=\pi\ii\sigma_j.
$$
Since $ \theta (\vec u)$ is 1-periodic for any variable $u_j$ then formulas (\ref{uln}), (\ref{RTF}), (\ref{ag}), (\ref{310}) give the final formula for the finite-gap solution of the Sine-Gordon equation:
\be\label{utheta}
u(x,t)= 2\ii\ln\frac{\theta(\vec\alpha x+\vec\beta t+\vec\gamma+\vec\delta)}{\theta(\vec\alpha x+\vec\beta t+\vec\gamma)} +C+2k\pi,\qquad k\in\mathbb{Z},
\ee
where $\vec\delta=(1/2,1/2,\ldots,1/2)$, and
$$
C=\ii\left(\sum\limits_{j=1}^N \int\limits_{a_j}\,\ln z\, d\omega_j(z)-
\ln(-1)^N\sqrt{p_0} \right).
$$
Later it was shown by A.R.Its \cite{Matv} that $C=0$.
Thus we prove
\begin{thm}
Let pairwise distinct complex numbers $E_j$ ($j=1,2,\ldots,2N$) and $\mu_l^0:=\mu_l(0,0)$ ($l=1,2,\ldots,N$) be given. Let there exists such a polynomial $f(z)$ with real coefficients that
\be\label{real}
\prod\limits_{j=1}^{2N}(z-E_j)-\prod\limits_{l=1}^{N}(z-\mu_l^0)(z-\overline{\mu_l^0})=zf^2(z).
\ee
Then defined by (\ref{utheta}) the function $u(x,t)$ is a real, finite-gap, infinitely differentiable for all  $x$ and $t$ solution of the Sine-Gordon equation (\ref{SG}).
\end{thm}
\begin{rem}
  In general, formula (\ref{utheta}) is true for complex valued solutions of (\ref{SG}). Polynomial relation (\ref{real}) gives necessary and sufficient conditions for the function (\ref{utheta}) to be real solution of the Sine-Gordon equation.
\end{rem}
\begin{rem}
  It is easy to see that solution (\ref{utheta}) depends in fact  on $\xi=(x+t)/2$ and $\eta=(x-t)/2$. Therefore the function $v(\xi,\eta)=u(\xi+\eta,\xi-\eta)$ is a solution of the  Sine-Gordon equation in light-cone variables:
  $$
  v_{\xi\eta}=\sin v.
  $$
\end{rem}

Now we return to the finite-gap potentials of the $L$-operator. Such a potentials are generated by $\exp(\ii u(x,t)/2)$ and $w(x,t)$. We have obtained the following representations:
$$
e^{\ii u(x.t)/2}=\frac{\theta(\vec\alpha x+\vec\beta t+\vec\gamma)}{\theta(\vec\alpha x+\vec\beta t+\vec\gamma+\vec\delta)}e^{\ii C/2}(-1)^k, \qquad k=1,2;
$$
$$
w(x,t)=2\ii\left(\frac{\partial}{\partial x}+\frac{\partial}{\partial t}\right)
\ln\frac{\theta(\vec\alpha x+\vec\beta t+\vec\gamma+\vec\delta)}{\theta(\vec\alpha x+\vec\beta t+\vec\gamma)},
$$
where $t$   must be viewed as a parameter. The properties of the theta function and the results, obtained above,  show  that $\exp(\ii u(x,t)/2)$ and $w(x,t)$ are bounded and almost periodic functions in $x$ and $t$.

\textbf{Acknowledgment}
The authors are grateful to Professor Marchenko V.A.
for his attention to this work and valuable advices.\\

\end{document}